\documentclass[twoside]{aiml12}

\usepackage{aiml12macro}
\usepackage{amssymb}
\usepackage{proof}
\usepackage{latexsym}
\usepackage{amsmath}
\usepackage{url}
\usepackage{slashed}
\usepackage{stmaryrd}

\def\Mf{\mathfrak{M}}
\def\Fc{\mathcal{F}}
\def\Km{\mathrm{Km}}
\def\prod{\rightarrow}
\def\prods{\Rightarrow}
\def\Scal{\mathcal{S}}
\def\SKm{\mathrm{SKm}}
\def\DKm{\mathrm{DKm}}
\def\tup{t\!\uparrow}
\def\tdn{t\!\downarrow}

\def\Acal{\mathcal{A}}

\def\Pcal{\mathcal{P}}
\def\Rcal{\mathcal{R}}

\def\impl {\supset}
\newcommand{\trans}[1]{\stackrel{#1}{\longrightarrow}}
\newcommand{\next}[1]{\succ^{#1}}

\newcommand{\obox}[1]{{[#1]}}
\newcommand{\odia}[1]{{\langle #1 \rangle}}
\newcommand{\seq}[2]{{({#1})\{{#2}\}}}
\newcommand{\lneg}[1]{{{#1}^{\bot}}}

\def\idd{\mathit{id}_d}
\def\landd{\land_d}
\def\lord{\lor_d}

\newcommand{\boxrd}[1]{\obox{#1}_d}
\newcommand{\diardn}[1]{\odia{#1 \! \downarrow}}
\newcommand{\diarup}[1]{\odia{#1 \! \uparrow}}

\def\addtree{\ll}
\def\assign{\mathrel{{\mathrel{\mathop:}=}}}
\def\lastname{Tiu, Ianovski, Gor\'e}

\begin{document}

\begin{frontmatter}
\title{Grammar Logics in Nested Sequent Calculus: Proof Theory and Decision Procedures}
\author{Alwen Tiu}
\address{Research School of Computer Science, The Australian National University}
\author{Egor Ianovski}
\address{Department of Computer Science, University of Auckland}
\author{Rajeev Gor\'e}
\address{Research School of Computer Science, The Australian National University}

\begin{abstract}
A grammar logic refers to an extension to the multi-modal logic K in
which the modal axioms are generated from a formal grammar.  We
consider a proof theory, in nested sequent calculus, of grammar logics
with converse, i.e., every modal operator $\obox{a}$ comes with a
converse $\obox{a}^{-1}.$ Extending previous works on nested sequent
systems for tense logics, we show all grammar logics (with or without
converse) can be formalised in nested sequent calculi, where the
axioms are internalised in the calculi as structural rules. Syntactic
cut-elimination for these calculi is proved using a procedure similar
to that for display logics.  If the grammar is context-free, then one
can get rid of all structural rules, in favor of deep inference and
additional propagation rules.  We give a novel semi-decision procedure
for context-free grammar logics, using nested sequent calculus with
deep inference, and show that, in the case where the given
context-free grammar is regular, this procedure terminates. Unlike all
other existing decision procedures for regular grammar logics in the
literature, our procedure does not assume that a finite state
automaton encoding the axioms is given.
\end{abstract}


\end{frontmatter}

\section{Introduction}

A grammar logic refers to an extension of the multi-modal logic K in
which the modal axioms are generated from a formal
grammar. Thus
given a set $\Sigma$ of indices, and a grammar production rule 
as shown below left,  where each $a_i$ and $b_j$ are in
$\Sigma$, we extend K with the multi-modal axiom shown below right:
$$a_1 a_2 \cdots a_l \prod
b_1 b_2 \cdots b_r
\qquad\qquad\qquad
\obox{a_1} \obox{a_2} \cdots \obox{a_l} A \impl
\obox{b_1} \obox{b_2} \cdots \obox{b_r} A
$$
The logic is a context-free grammar logic if $l=1$ and furthermore, is
a right linear grammar logic if the production rules also define a
right linear grammar. The logic is a regular grammar logic if the set
of words generated from each $a \in \Sigma$ using the grammar
production rules is a regular language. A right linear grammar logic
is also a regular grammar logic since a right linear grammar can be
converted to a finite automaton in polynomial time. Adding 
``converse'' gives us alphabet symbols like $\bar{a}$ which
correspond to the converse modality $\obox{\bar a}$ and lead to
multi-modal extensions of tense logic Kt where each
modality 
$\obox{a}$ and its converse 
$\obox{\bar a}$
obey the interaction axioms 
$
A \impl \obox{a}\odia{\bar a} A$
and 
$A \impl \obox{\bar a}\odia{a} A.$

Display calculi~\cite{Belnap82JPL} can handle grammar logics with converse since 
they all fall into the primitive fragment identified by Kracht~\cite{Kra96}. Display calculi all enjoy
Belnap's general cut-elimination theorem, but it is well-known that they are not suitable for proof-search.
Our work is motivated by the problem of automating proof search for 
display calculus. As in our previous work~\cite{Gore09tableaux,Gore10AiML,Gore11LMCS}, 
we have chosen to work not directly in display calculus, but in a slightly different calculus based 
on nested sequents~\cite{kashima-cut-free-tense,Brunnler09AML}, which we call
shallow nested sequent calculi. The syntactic constructs of nested sequents
are closer to traditional sequent calculu, so as to allow us 
to use familiar notions in sequent calculus proof search procedures, such as  
the notions of saturation and loop checking, to automate proof search. 
A common feature of shallow nested sequent calculus and display calculus is 
the use display postulates and other complex structural rules.
These structural rules are the main obstacle to effective proof search,
and our (proof theoretic) methodology for designing proof search calculi is guided by the problem of eliminating
these structural rules entirely. 
We show here
how our methodology can be used to derive proof search calculi for context-free grammar logics. 

The general satisfiability problem for a grammar logic is to decide
the satisfiability of a given formula when given a set of production
rules or when given an explicit finite state automaton (FSA) for the
underlying grammar. 

Nguyen and Sza{\l}as~\cite{Nguyen11StudiaLogica} give an excellent
summary of what is known about this problem, as outlined next. Grammar
logics were introduced by del Cerro and
Penttonen~\cite{del-Cerro-Penttonen}. Baldoni et
al~\cite{BaldoniGM98} used prefixed tableaux
to show that this problem is decidable for right linear logics but is
undecidable for context free grammar
logics. Demri~\cite{Demri01} used an embedding into
propositional dynamic logic with converse to prove this problem is
EXPTIME-complete for right linear logics. 
Demri and de Nivelle~\cite{Demri05} gave an embedding of
the satisfiability problem for regular grammar logics into the
two-variable guarded fragment of first-order logic and showed that
satisfiability of regular grammar logics with converse is also
EXPTIME-complete. 
Seen as description logics with inverse roles and complex role inclusions,
decision procedures for regular grammar logics have also been studied
extensively by Horrocks, et. al., see, e.g., \cite{Horrocks04AI,Horrocks06KR,Kazakov08KR}.
Gor\'e and Nguyen~\cite{GoreN05} gave an EXPTIME
tableau decision procedure for the satisfiability of regular grammar
logics using formulae labelled with automata states.  Finally, Nguyen
and Sza{\l}as~\cite{Nguyen09CADE,Nguyen11StudiaLogica} gave
an extension of this method to handle converse by using the cut rule.
In an unpublished manuscript, Nguyen has
shown how to use the techniques of Gor\'e and
Widmann~\cite{GoreW10} to avoid the use of the cut rule. 
But as far as we know, there is no
comprehensive sequent-style proof theory for grammar logics with
converse which enjoys a syntactic cut-elimination theorem and which is
amenable to proof-search.

We consider a proof theory, in nested sequent calculus,
of grammar logics with converse, i.e., every modal operator $\obox{a}$
comes with a converse $\obox{a}^{-1}.$ Extending previous works on
nested sequent systems for (bi-)modal logics~\cite{Gore09tableaux,Gore11LMCS}, we show,
in Section~\ref{sec:skm}, that all grammar
logics (with or without converse) can be formalised in (shallow)
nested sequent calculi, where the axioms are internalised in the
calculi as structural rules. Syntactic cut-elimination for these
calculi is proved using a procedure similar to that for display
logics. We then show, in Section~\ref{sec:dkm}, that if the grammar is context-free, then one can get rid of all
structural rules, in favor of deep inference and additional 
propagation rules.

We then recast the problem of deciding grammar logics for the specific
cases where the grammars are regular, using nested sequent calculus
with deep inference.  We first give, in Section~\ref{sec:auto-proc}, 
a decision procedure in the case where the regular grammar is given in the form of a FSA. 
This procedure is similar to existing tableaux-based 
decision procedures~\cite{Horrocks06KR,Nguyen09CADE,Nguyen11StudiaLogica}, 
where the states and transitions of the FSA is incorporated into proof rules for propagation
of diamond-formulae. 
This procedure serves as a stepping stone to defining the more general decision procedure 
which does not depend
on an explicit representation of axioms as a FSA in Section~\ref{sec:grammar-proc}. 
The procedure in Section~\ref{sec:grammar-proc} is actually a semi-decision procedure that 
works on any finite set of context-free grammar axioms. However, we show that, in the case where the
given grammar is regular, this procedure terminates. The procedure avoids the requirement 
to provide a FSA for the given axioms. This is significantly different from existing decision procedures
for regular grammar logics~\cite{Demri05,GoreN05,Nguyen11StudiaLogica,Nguyen09CADE}, 
where it is assumed that a FSA encoding the axioms of the logics is given. 

In this work, we follow Demri and de Nivelle's presentation of grammar axioms as a 
semi-Thue system~\cite{Demri05}. 
The problem of deciding whether a context-free semi-Thue system is regular or not appears to be 
still open; see \cite{Kazakov08KR} for a discussion on this matter.  
Termination of our generic procedure for regular grammar logics 
of course does not imply solvability of this problem
as it is dependent on the assumption that the given grammar is regular 
(see Theorem~\ref{thm:grammar-proc-terminates}).

\section{Grammar logics}
\label{sec:logic}

The language of a multi-modal logic is defined w.r.t. to an alphabet
$\Sigma$, used to index the modal operators.  We use $a, b$ and $c$,
possibly with subscripts, for elements of $\Sigma$ and use $u$ and $v$,
for elements of $\Sigma^*$, the set of finite strings over $\Sigma$. 
We use $\epsilon$ for the empty string.
We define an operation $\bar .$ (converse) on alphabets to capture
converse modalities following Demri~\cite{Demri05}. The converse
operation satisfies $\bar {\bar a} = a.$ We assume that $\Sigma$ can
be partitioned into two distinct sets $\Sigma^+$ and $\Sigma^-$ such that 
$a \in \Sigma^+$ iff $\bar a \in \Sigma^{-}.$ The converse operation is extended to
strings in $\Sigma^*$ as follows: if $u = a_1a_2 \ldots a_n$, then
$
\bar u = \bar a_n \bar a_{n-1} \ldots \bar a_2 \bar a_1,
$
where $n \geq 0.$ Note that if $u = \epsilon$ then $\bar u = \epsilon.$

We assume a given denumerable set of atomic formulae, ranged over
by $p$, $q$, and $r.$
The language of formulae is given by the following, where
$a \in \Sigma$:
$$
A ::= p \mid \neg A \mid A \lor A \mid A \land A \mid \obox a A \mid \odia a A
$$
Given a formula $A$, we write $\lneg A$
for the negation normal form (nnf) of $\neg A.$ Implication $A \impl B$ is
defined as $\neg A \lor B.$

\begin{definition}
A {\em $\Sigma$-frame} is a pair $\langle W, R\rangle$
of a non-empty set of worlds and a set of binary relations $\{ R_a \}_{a\in \Sigma}$ 
over $W$ satisfying, for every $a \in \Sigma$, 
$
R_a = \{(x,y) \mid R_{\bar a}(y,x) \}.
$
A {\em valuation} $V$ is a mapping from propositional variables to sets of worlds. 
A {\em model} $\Mf$ is a triple $\langle  W,R,V\rangle$ where $\langle W,R \rangle$
is a frame and $V$ is a valuation. 
The relation $\models$ is defined inductively as follows: 
\begin{itemize}
\item $\Mf, x \models p$ iff $x \in V(p).$
\item $\Mf, x \models \neg A$ iff  $\Mf, x \not \models A$.
\item $\Mf, x \models A \land B$ iff $\Mf, x \models A$ and $\Mf, x \models B.$
\item $\Mf, x \models A \lor B$ iff $\Mf, x \models A$ or $\Mf, x \models B.$
\item For every $a \in \Sigma$, 
$\Mf, x \models \obox a A$ iff for every $y$ such that $R_a(x,y)$, 
$\Mf, y \models A.$
\item For every $a \in \Sigma$, 
$\Mf, x \models \odia a A$ iff there exists $y$ such that $R_a(x,y),$
$\Mf, y \models A.$
\end{itemize}
A formula $A$ is {\em satisfiable} iff there exists a $\Sigma$-model
$\Mf = \langle W, R, V \rangle$ and a world $x \in W$ such that 
$\Mf, x \models A.$
\end{definition}

We now define a class of multi-modal logics, given $\Sigma$, that is
induced by {\em production rules} for strings from $\Sigma^*.$ We follow the
framework in \cite{Demri05}, using semi-Thue systems to define the logics. 
A production rule is a binary relation over strings in $\Sigma^*$, interpreted
as a rewrite rule on strings. We use the notation $u \prod v$ to denote 
a production rule which rewrites $u$ to $v.$
A {\em semi-Thue} system is a set  $S$ of production rules. It is
{\em closed}  if $u \prod v \in S$ implies $\bar u \prod \bar v \in S.$ 

Given a $\Sigma$-frame $\langle W, R\rangle$, 
we define another family of accessibility relations indexed by $\Sigma^*$
as follows: 
$R_\epsilon = \{(x,x) \in x \in W\}$ and 
for every $u\in \Sigma^*$ and for every $a \in \Sigma$, 
$R_{ua} = \{(x,y) \mid (x,z) \in R_u, (z,y)\in R_a, \mbox{ for some $z \in W$} \}.$

\begin{definition}
\label{def:S-frame}
Let $u \prod v$ be a production rule and let $\Fc = \langle W, R \rangle$ be a
$\Sigma$-frame. $\Fc$ is said to satisfy $u \prod v$ if 
$R_v \subseteq R_u.$ $\Fc$ satisfies a semi-Thue system $S$ if it
satisfies every production rule in $S.$
\end{definition}

\begin{definition}
Let $S$ be a semi-Thue system. A formula $A$ is said to be 
{\em $S$-satisfiable} iff there is a model $\Mf = \langle W , R , V\rangle$
such that $\langle W, R\rangle$ satisfies $S$ and 
$\Mf, x \models A$ for some $x \in W.$
$A$ is said to be {\em $S$-valid} if for every $\Sigma$-model 
$\Mf = \langle W, R, V\rangle$ that satisfies $S$, we have $\Mf, x \models A$
for every $x \in W.$
\end{definition}

Given a string $u = a_1 a_2 \ldots a_n$ and a formula $A$, we write $\odia u A$
for the formula $\odia {a_1} \odia {a_2} \cdots \odia{a_n} A.$
The notation $\obox u A$ is defined analogously.
If $u = \epsilon$ then $\odia u A = \obox u A = A.$

\begin{definition}
Let $S$ be a closed semi-Thue system over an alphabet $\Sigma$. The system $\Km(S)$
is an extension of the standard Hilbert system for multi-modal $\Km$ 
(see, e.g., \cite{blackburn07handbook})
with the following axioms: 
\begin{itemize}
\item for each $a \in \Sigma$, a {\em residuation axiom:}
$
A \impl \obox a \odia {\bar a} A
$
\item and for each $u \prod v \in S$, an axiom
$
\obox u A \impl \obox v A.
$
\end{itemize}
\end{definition}
Note that because we assume that $S$ is closed, each 
axiom $\obox u A \impl \obox v A$ has an {\em inverted version}
$\obox {\bar u} A \impl \obox {\bar v} A.$

The following theorem can be proved following a similar 
soundness and completeness proof for Hilbert systems for 
modal logics (see, e.g., \cite{blackburn07handbook}). 
\begin{theorem}
\label{thm:Km-S}
A formula $F$ is $S$-valid iff $F$ is provable in $\Km(S).$
\end{theorem}

\section{Nested sequent calculi with shallow inference}
\label{sec:skm}

We now give a sequent calculus for $\Km(S)$, by using the
framework of nested sequent calculus~\cite{kashima-cut-free-tense,Brunnler09AML,Gore09tableaux,Gore11LMCS}. 
We follow the notation used in \cite{kashima-cut-free-tense,Gore11LMCS}, 
extended to the multi-modal case. From this section onwards, 
we shall be concerned only with formulae in nnf, so we can restrict to one-sided
sequents. 

A nested sequent is a multiset of the form shown below at left
$$
A_1,\dots,A_m, \seq {a_1} {\Delta_1} ,\dots, \seq {a_n} {\Delta_n}
\quad\quad\quad
A_1 \lor \cdots \lor A_m \lor \obox {a_1} B_1 
\lor \cdots \lor \obox{a_n} B_n
$$
where each $A_i$ is a formula and each $\Delta_i$ is a nested sequent. 
The structural connective $\seq {a} {.}$ is a proxy for the modality
 $\obox {a}$,
so this nested sequent can be interpreted as the formula shown above
right (modulo associativity and commutativity of $\lor$),
where each $B_i$ is the interpretation of $\Delta_i.$
We shall write $\seq u \Delta$, where $u = a_1 \cdots a_n \in \Sigma^*$, to denote the structure:
$$
\seq {a_1} {\seq {a_2} {\cdots \seq {a_n} \Delta} \cdots }.
$$
A {\em context} is a nested sequent with a `hole' $[~]$ in place of a
formula: this notation should not be confused with the modality $\obox
a$. We use $\Gamma[~]$, $\Delta[~]$, etc.\ for contexts.
Given a context $\Gamma[~]$ and a nested sequent $\Delta$, we write
$\Gamma[\Delta]$ to denote the nested sequent obtained by replacing the hole
in $\Gamma[~]$ with $\Delta.$ 

The core inference rules for multi-modal $\SKm$ (without axioms)
are given in Figure~\ref{fig:SKm}. The rule $r$ is called a {\em residuation rule}
(or display postulate) and corresponds to the residuation axioms.

\begin{figure}[t]
  \begin{tabular}[c]{cc}
  \begin{tabular}[c]{ccccc}
 $\infer[\mathit{id}]{\Gamma, p, \neg p}{}$ &
 $\infer[\mathit{cut}]{\Gamma, \Delta}{\Gamma, A & \Delta, \lneg A}$
 & $\infer[\mathit{ctr}]{\Gamma, \Delta}{\Gamma, \Delta, \Delta}$
 & $\infer[\mathit{wk}]{\Gamma, \Delta}{\Gamma}$
 & $\infer[\mathit{r}]{\seq {\bar a} \Gamma, \Delta}{\Gamma, \seq a \Delta}$
  \end{tabular}
\\[2em]
  \begin{tabular}[c]{cccc}
  $\infer[\land]{\Gamma, A \land B}{\Gamma, A & \Gamma, B}$
 & $\infer[\lor]{\Gamma, A \lor B}{\Gamma, A, B}$
 & $\infer[\obox a]{\Gamma, \obox a A}{\Gamma, \seq a A}$
 & $\infer[\odia a]{\Gamma, \seq a \Delta, \odia a A}
                   {\Gamma, \seq a {\Delta, A}}$
  \end{tabular}
  \end{tabular}
\caption{The inference rules of $\SKm$}
\label{fig:SKm}
\end{figure}

To capture $\Km(S)$, we need to convert each axiom generated from $S$ to
an inference rule. 
Each production rule $u \prod v$ gives rise to the 
axiom $\obox u A \impl \obox v A$, or equivalently,
$\odia {\bar v} A \impl \odia {\bar u} A.$
The latter is an instance of the Kracht's {\em primitive axioms}~\cite{Kra96} (generalised
to the multimodal case).  
Thus, we can convert the axiom into a structural rule
following Kracht's rule scheme for primitive axioms:
$$\infer[]{\seq v {\Delta}, \Gamma}{\seq u {\Delta}, \Gamma}$$
Let $\rho(S)$ be the set of structural rules induced by the
semi-Thue system $S.$

\begin{definition}
Let $S$ be a closed semi-Thue system $S$ over an alphabet $\Sigma$. 
$\SKm(S)$ is the proof system obtained by
extending $\SKm$ with $\rho(S).$
\end{definition}

We say that two proof systems are equivalent if and only if
they prove the same set of formulae. 
\begin{theorem}
\label{thm:SKm-Km-equiv}
The system $\SKm(S)$ and $\Km(S)$ are equivalent. 
\end{theorem}
The cut-elimination proof for $\SKm(S)$ follows a similar generic procedure
for display calculi~\cite{Belnap82JPL,Kra96}, which has been 
adapted to nested sequent in~\cite{Gore11LMCS}.
The key to cut-elimination is to show that $\SKm(S)$ has the {\em display property}.

\begin{lemma}
\label{lm:display}
Let $\Gamma[\Delta]$ be a nested sequent. Then there exists a nested sequent $\Gamma'$
such that $\Gamma[\Delta]$ is derivable from the nested sequent $\Gamma',\Delta$, and vice versa,
using only the residuation rule $r.$
\end{lemma}

\begin{theorem}
Cut elimination holds for $\SKm(S).$
\end{theorem}
\begin{proof}
This is a straightforward adaptation of the cut-elimination proof in \cite{Gore11LMCS} for tense logic. 
\end{proof}

\section{Deep inference calculi}
\label{sec:dkm}

Although the shallow system $\SKm(S)$ enjoys cut-elimination, 
proof search in its cut-free fragment is difficult to automate,
due to the presence of structural rules. 
To reduce the non-determinism caused by structural rules, we consider
next a proof system in which all structural rules (including those
induced by grammar axioms) can be absorbed into logical rules. 
As the display property in Lemma~\ref{lm:display} suggests, the residuation rule
allows one to essentially apply an inference rule to a particular
subsequent nested inside a nested sequent, by displaying 
that subsequent to the top and undisplaying it back to its original
position in the nested sequent. 
It is therefore quite intuitive that one way to get rid of 
the residuation rule is to allow {\em deep inference} rules, that apply
deeply within any arbitrary context in a nested sequent.

The deep inference system $\DKm$, which corresponds to $\SKm$, is given
in Figure~\ref{fig:DKm}. As can be readily seen, the residuation rule
is absent and contraction and weakening are absorbed into logical rules. 

\begin{figure}[t]
$$
\begin{array}{c@{\quad}c@{\quad}c}
\infer[\idd]
{\Gamma[p,\neg p]}{}
&
\infer[\landd]
{\Gamma[A \land B]}
{\Gamma[A \land B, A] & \Gamma[A \land B, B]}
&
\infer[\lord]
{\Gamma[A \lor B]}
{\Gamma[A \lor B, A, B]}
\\ \\ 
\infer[\boxrd a]
{\Gamma[\obox a A]}{\Gamma[\obox a A, \seq a A]}
&
\infer[\diarup{a}]
{\Gamma[\seq a {\Delta}, \odia a A]}
{\Gamma[\seq a {\Delta, A}, \odia a A]}
&
\infer[\diardn{a}]
{\Gamma[\seq a {\Delta, \odia {\bar a} A}]}
{\Gamma[\seq a {\Delta, \odia {\bar a} A}, A]}
\end{array}
$$
\caption{The inference rules of $\DKm$}
\label{fig:DKm}
\end{figure}

To fully absorb the residuation rule, and other structural rules
induced by production rules, we need to modify the introduction
rules for diamond-formulae. Their introduction rules will be
dependent on what axioms one assumes. We refer to these
introduction rules for diamond-formulae as {\em propagation rules}. 
This will be explained shortly, but first we need to define
a couple of notions needed to define propagation rules.

Let $S$  be a closed semi-Thue system over alphabet $\Sigma$. 
We write $u \prods_S v$
to mean that the string
$v$
can be reached from $u$ by applying the production rules
(as rewrite rules) in $S$ successively to $u.$
Define $L_a(S) = \{ u \mid a \prods_S u \}.$ Then $L_a(S)$ defines
a language generated from $S$ with the start symbol $a.$

A nested sequent can be seen as a tree whose nodes are multisets of
formulae, and whose edges are labeled with elements of $\Sigma.$ We
assume that each node in a nested sequent can be identified uniquely,
i.e., we can consider each node as labeled with a unique position
identifier. An internal node of a nested sequent is a node which is
not a leaf node. We write $\Gamma[~]_i$ to denote a context in which the
hole is located in the node at position $i$ in the tree 
representing $\Gamma[~].$ This generalises to multi-contexts, so
$\Gamma[~]_i[~]_j$ denotes a two-hole context, one hole located at $i$
and the other at $j$ (they can be the same location).  From now on, we
shall often identify a nested sequent with its tree representation, so
when we speak of a node in $\Gamma$, we mean a node in the tree
of $\Gamma.$ If $i$ and $j$ are nodes in $\Gamma$, we write $i \next{a}
j$ when $j$ is a child node of $i$ and the edge from $i$ to $j$ is
labeled with $a.$
If $i$ is a node in the tree of $\Gamma$, we write
$\Gamma|i$ to denote the multiset of formula occuring
in the node $i.$
Let $\Delta$ and $\Gamma$ be nested sequents.
Suppose $i$ is a node in $\Gamma$. 
Then we write 
$\Gamma(i \addtree \Delta)$
for the nested sequent obtained from $\Gamma$ by adding 
$\Delta$ to node $i$ in $\Gamma.$ Note that for such an addition to preserve
the uniqueness of the position identifiers of the resulting tree, we need to
rename the identifiers in $\Delta$ to avoid clashes. We shall assume implicitly that such
a renaming is carried out when we perform this addition.

\begin{definition}[Propagation automaton.]
  A propagation automaton is a finite state automaton
  $\Pcal = (\Sigma,Q, I, F, \delta)$ where $Q$ is a finite set of
  states, $I = \{s\}$ is a singleton set of initial state and $F = \{t\}$ is a 
  singleton set of final
  state with $s,t \in Q$, and for every $i,j \in Q$, if $i
  \trans{a} j \in \delta$ then $j \trans{\bar a} i \in \delta.$
\end{definition}
In other words, a propagation automaton is just a finite state automaton (FSA)
where each transition has a dual transition. 

\begin{definition}
Let $\Acal = (\Sigma, Q, I, F, \delta)$ be a 
FSA. 
Let $\vec i = i_1,\dots,i_n$ and $\vec j = j_1,\dots,j_n$ be two sequences of states in $Q.$
Let $[i_1 \assign j_1,\dots, i_n \assign j_n]$ (we shall abbreviate this as $[\vec i \assign \vec j]$)
be a (postfix) mapping from $Q$ to $Q$ that maps 
$i_m$ to $j_m$, where $1 \leq m \leq n$, and is the identity map otherwise. 
This mapping is extended to a (postfix) mapping between sets of states 
as follows: given $Q' \subseteq Q$,
$Q'[\vec i \assign \vec j] = \{k[\vec i \assign \vec j] \mid k \in Q'\}.$
The automaton $\Acal[\vec i \assign \vec j]$ is the 
tuple $(\Sigma, Q[\vec i \assign \vec j], I[\vec i \assign \vec j], F[\vec i \assign \vec j], \delta')$
where 
$$
\delta' = \{k[\vec i \assign \vec j] \trans{a} l[\vec i \assign \vec j] \mid k \trans{a} l \in \delta \}.
$$
\end{definition}

To each nested sequent $\Gamma$, and nodes $i$ and $j$ in $\Gamma$,
we associate a propagation automaton $\Rcal(\Gamma,i,j)$ 
as follows: 
\begin{enumerate}
\item the states of $\Rcal(\Gamma,i,j)$ are the nodes of (the tree of) $\Gamma$;
\item $i$ is the initial state of $\Rcal(\Gamma,i,j)$ and $j$ is its final
  state;
\item each edge $x \next{a} y$ in $\Gamma$ corresponds to two transitions in  
$\Rcal(\Gamma,i,j)$: 
$
x \trans{a} y$
and
$y \trans{\bar a} x.$
\end{enumerate}

Note that although propagation automata are defined for nested sequents, they can be
similarly defined for (multi-)contexts as well, as contexts are just sequents containing 
a special symbol $[~]$ denoting a hole. So in the following, we shall often treat a 
context as though it is a nested sequent. 

A semi-Thue system $S$ over alphabet $\Sigma$ is {\em context-free} if its
production rules
are all of the form $a \to u$ for some $a \in \Sigma.$

In the following, to simplify presentation, we shall use the same
notation to refer to an automaton $\Acal$ 
and the regular language it accepts.
Given a context-free closed semi-Thue system $S$, the {\em propagation rules for
$S$} are all the rules of the following form
where
$i$ and $j$ are two  (not necessarily distinct) nodes of $\Gamma$:
$$
\infer[p_S, \mbox{ provided } {\Rcal(\Gamma[~]_i[~]_j, i, j) \cap L_a(S) \not = \emptyset.}]
{\Gamma[\odia a A]_i[\emptyset]_j}
{\Gamma[\odia a A]_i[A]_j}
$$ 
Note that the intersection of a regular language and 
a context-free language is a context-free language (see, e.g., Chapter 3 in \cite{ginsburg} for
a construction of the intersection),
and since the emptiness checking for context-free languages
is decidable~\cite{ginsburg}, the rule $p_S$ can be effectively
mechanised. 

\begin{definition}
Given a context-free closed semi-Thue system $S$ over
an alphabet $\Sigma$,
the proof system $\DKm(S)$ is  obtained by
extending $\DKm$ with $p_S.$
\end{definition}

We now show that $\DKm(S)$ is equivalent to $\SKm(S).$
The proof relies on a series of lemmas showing admissibility of
all structural rules of $\SKm(S)$ in $\DKm(S).$ The proof follows
the same outline as in the case for tense logic~\cite{Gore11LMCS}. 
The adaptation of the proof in \cite{Gore11LMCS} is quite straightforward,
so we shall not go into detailed proofs but instead just outline the required 
lemmas. Some of their proofs are outlined in the appendix. In the following
lemmas, we shall assume that $S$ is a closed context-free semi-Thue system over
some $\Sigma.$

Given a derivation $\Pi$, we denote with $|\Pi|$ the {\em height} of $\Pi$, i.e.,
the length (i.e., the number of edges) of the longest branch in $\Pi.$
A rule $\rho$ is said to be {\em admissible} in $\DKm(S)$ if 
provability of its premise(s) in $\DKm(S)$ implies provability of
its conclusion in $\DKm(S).$ 
It is {\em height-preserving admissible} if whenever the premise has
a derivation then the conclusion has a derivation of the same height,
in $\DKm(S).$

Admissibility of the weakening rule is a consequence of the following lemma.
\begin{lemma}
\label{lm:weak}
Let $\Pi$ be a derivation of $\Gamma[\emptyset]$ in $\DKm(S).$
Then there exists a derivation $\Pi'$ of $\Gamma[\Delta]$ in $\DKm(S)$ 
such that $|\Pi| = |\Pi'|.$
\end{lemma}

The admissibility proofs of the remaining structural rules all follow the same pattern: the most
important property to prove is that, if a propagation path for a
diamond formula exists between two nodes in the premise, 
then there exists a propagation path for the same formula, between the same nodes,
in the conclusion of the rule.

\begin{lemma}
\label{lm:adm-r}
The rule $r$ is height-preserving admissible in $\DKm(S).$
\end{lemma}
Admissibility of contraction is proved indirectly by
showing that it can be replaced by
a formula contraction rule and a distributivity rule:
$$
\infer[\mathit{actr}]
{\Gamma[A]}
{\Gamma[A,A]}
\qquad
\infer[\mathit{m}]
{\Gamma[\seq a {\Delta_1,\Delta_2}]}
{\Gamma[\seq a {\Delta_1}, \seq a {\Delta_2}]}
$$
The rule $m$ is also called a {\em medial} rule and is typically used
to show admissibility of contraction in deep inference~\cite{Brunnler01LPAR}. 

\begin{lemma}
The rule $\mathit{ctr}$ is admissible in 
$\DKm(S)$ plus $\mathit{actr}$ and $m.$
\end{lemma}

\begin{lemma}
\label{lm:adm-medial}
The rules $\mathit{actr}$ and $m$ are height-preserving admissible in $\DKm(S)$. 
\end{lemma}
Admissibility of contraction then follows immediately. 
\begin{lemma}
The contraction rule $\mathit{ctr}$ is admissible in $\DKm(S).$
\end{lemma}

\begin{lemma}
\label{lm:adm-rho-S}
The structural rules $\rho(S)$ of $\SKm(S)$ are height-preserving 
admissible in $\DKm(S).$
\end{lemma}

\begin{theorem}
\label{thm:SKm-DKm-equiv}
For every context-free closed semi-Thue system $S$, 
the proof systems $\SKm(S)$ and $\DKm(S)$ are equivalent.
\end{theorem}

\section{Regular grammar logics}

A context free semi-Thue system $S$ over $\Sigma$ 
is regular if for every $a \in \Sigma$, the language $L_a(S)$ is
a regular language.

In this section, we consider logics generated by regular closed semi-Thue
systems. We assume in this case that the union of the regular languages $\{L_a(S) \mid a \in \Sigma \}$
is represented explicitly as an FSA $\Acal$ with no silent transitions.
Thus
$
\Acal = (\Sigma, Q, I, F, \delta)
$
where $Q$ is a finite set of states, 
$I \subseteq Q$ is the set of initial states, $F \subseteq Q$ is the set
of final states, and $\delta$ is the transition relation. 
Given $\Acal$ as above, we write $s \trans{a}_\Acal t$ to mean $s \trans{a} t
\in \delta.$
We further assume that each $a \in \Sigma$ has a unique initial state $init_a \in I.$

We shall now define an alternative deep inference system given this
explicit representation of the grammar axioms as an FSA. 
Following similar tableaux systems in the literature that utilise such an
automaton representation~\cite{Horrocks06KR,Nguyen09CADE,Nguyen11StudiaLogica}, 
we use the states of the FSA to index formulae in
a nested sequent to record stages of a propagation. 
For this, we first introduce a form of labeled formula, written $s : A$,
where $s \in Q.$  The propagation rules corresponding to 
$\Acal$ are: 
$$
\infer[i]
{\Gamma[\odia a A]}
{
 \Gamma[\odia a A, init_a : A]
}
\qquad\qquad\qquad
\infer[\tup, \hbox{ if } s \trans{a}_\Acal s']
{\Gamma[s : A, \seq a {\Delta}]}
{
 \Gamma[s : A, \seq a {s' : A, \Delta}]
}
$$
$$
\infer[f, ~ \hbox{if } s \in F]
{\Gamma[s : A]}
{\Gamma[s : A, A]}
\qquad\qquad\qquad
\infer[\tdn, \hbox{ if } s \trans{\bar a}_\Acal s'. ]
{\Gamma[\seq a {s: A, \Delta}]}
{\Gamma[\seq a {s: A, \Delta}, s':A]}
$$

\begin{definition}
Let $S$ be a regular closed semi-Thue system over $\Sigma$ 
and let $\Acal$ 
be an FSA representing the regular language 
generated by $S$ and $\Sigma.$
$\DKm(\Acal)$ is the proof system $\DKm$ extended
with the rules $\{i, f, \tdn, \tup \}$ for $\Acal.$
\end{definition}

It is intuitively clear that $\DKm(\Acal)$ and $\DKm(S)$ are equivalent,
when $\Acal$ defines the same language as $L(S).$ Essentially, a propagation
rule in $\DKm(S)$ can be simulated by $\DKm(\Acal)$ using one or more
propagations of labeled formulae. The other direction follows from the fact
that when a diamond formula $\odia a A$ is propagated, via the use of labeled
formulae, to a labeled formula $s : A$ where $s$ is a final state, then there must
be a chain of transitions between labeled formulae for $A$ 
whose string forms an element of $\Acal$, hence also in $L_a(S).$ One can
then propagate directly $\odia a A$ in $\DKm(S).$
\begin{theorem}
\label{thm:DKm-A-eq-S}
Let $S$ be a regular closed semi-Thue system over $\Sigma$ and let $\Acal$ be a 
FSA representing the regular language generated by $S$ and $\Sigma.$
Then $\DKm(S)$ and $\DKm(\Acal)$ are equivalent.
\end{theorem}

\section{Decision procedures}

We now show how the proof systems $\DKm(\Acal)$ and $\DKm(S)$ can
be turned into decision procedures for regular grammar logics. Our aim is
to derive the decision procedure for $\DKm(S)$ directly without the need
to convert $S$ explicitly to an automaton; the decision procedure $\DKm(\Acal)$
will serve as a stepping stone towards this aim. 
The decision procedure for $\DKm(S)$ is a departure from all existing
decision procedures for regular grammar 
logics (with or without converse)~\cite{Horrocks06KR,Demri05,GoreN05,Nguyen09CADE,Nguyen11StudiaLogica}
that assume that an FSA representing $S$ is given.

\subsection{An automata-based procedure}
\label{sec:auto-proc}

The decision procedure for $\DKm(\Acal)$ is basically just backward proof search,
where one tries to saturate each sequent in the tree of sequents until 
either the $\idd$ rule is applicable, or a certain stable state is reached. 
When the latter is reached, we show that a counter model to the
original nested sequent can be constructed. Although we obtain this procedure via
a different route, the end result is very similar to the tableaux-based decision procedure 
in \cite{Horrocks06KR}. 
In particular, our notion of a stable state (see the definition of $\Acal$-stability below) used 
to block proof search is the same as the blocking condition in 
tableaux systems~\cite{Horrocks06KR,Demri05,GoreN05,Nguyen11StudiaLogica,Nguyen09CADE}, which
takes advantange of the labeling of formulae with the states of the automaton. 

\begin{figure}[t]
$Prove_1(\Acal, \Gamma)$
\begin{enumerate}
\item If $\Gamma = \Gamma'[p,\neg p]$, return $\top$.
\item If $\Gamma$ is $\Acal$-stable, return $\bot.$
\item If $\Gamma$ is not saturated: 
  \begin{enumerate}
\item If $A \lor B \in \Gamma|i$ but $A \notin \Gamma|i$
or $B \notin \Gamma|i$, then let $\Gamma' \assign \Gamma(i \addtree {\{A,B\}})$
and return $Prove_1(\Acal,\Gamma').$

\item Suppose $A_1 \land A_2 \in \Gamma|i$ but neither $A_1 \in \Gamma|i$
nor $A_2 \in \Gamma|i$. Let $\Gamma_1 = \Gamma(i \addtree \{A_1\})$
and $\Gamma_2 = \Gamma(i \addtree \{A_2\}).$
Then return $\bot$ if $Prove_1(\Acal,\Gamma_j) = \bot$ for some $j \in \{1,2\}.$
Otherwise return $\top.$
  \end{enumerate}

\item If $\Gamma$ is not $\Acal$-propagated: 
then there is a node $i$ s.t. one of the following applies:
\begin{enumerate}
\item $\odia a A \in \Gamma|i$ but $init_a : A \not \in \Gamma|i$. 
Then let $\Gamma' := \Gamma(i \addtree \{init_a : A\}).$
\item $s : A \in \Gamma|i$ and $s \in F$, but $A \not \in \Gamma|i.$
Then let $\Gamma' := \Gamma(i \addtree \{A\}).$

\item $s : A \in \Gamma|i$, there is $j$ s.t.
$i \next{a} j$ and $s \trans{a}_\Acal t$, but $t : A \not \in \Gamma|j.$
Then let $\Gamma' := \Gamma(j \addtree \{t:A\}).$

\item $s : A \in \Gamma|i$, there is $j$ s.t.
$j \next{a} i$ and $s \trans{\bar a}_\Acal t$, but $t : A \not \in \Gamma|j.$
Then let $\Gamma' := \Gamma(j \addtree \{t:A\}).$
\end{enumerate}
Return $Prove_1(\Acal,\Gamma').$

\item If there is an internal node $i$ in $\Gamma$ that is not realised:
Then there is $\obox a A \in \Gamma|i$ such that $A \not \in \Gamma|j$
for every $j$ s.t. $i \next{a} j.$ 
Let $\Gamma' := \Gamma(i \addtree \seq a {A}).$
Return $Prove_1(\Acal,\Gamma').$

\item If there is a leaf node $i$ that is not realised
and is not a loop node: 
Then there is $\obox a A \in \Gamma|i$.
Let $\Gamma' := \Gamma(i \addtree \seq a {A}).$
Return $Prove_1(\Acal,\Gamma').$
\end{enumerate}
\caption{An automata-based prove procedure.}
\label{fig:prove1}
\end{figure}

\begin{definition}[Saturation and realisation]
A node $i$ in $\Gamma$ is {\em saturated} if the following hold:
\begin{enumerate}
\item If $A \in \Gamma|i$ then $\lneg A \not \in \Gamma|i.$ 
\item If $A \lor B \in \Gamma|i$ then $A \in \Gamma|i$ and $B \in \Gamma|i$.
\item If $A \land B \in \Gamma|i$ then $A \in \Gamma|i$ or $B \in \Gamma|i$.
\end{enumerate}
$\Gamma|i$ is {\em realised} if 
$\obox a A \in \Gamma|i$ implies that there exists $j$ such that $i \next{a} j$
and $A \in \Gamma|j.$
\end{definition}

\begin{definition}[$\Acal$-propagation]
Let $\Acal = (\Sigma, Q, I, F, \delta)$. A nested sequent
$\Gamma$ is said to be {$\Acal$-propagated} if for every 
node $i$ in $\Gamma$, the following hold: 
\begin{enumerate}
\item If $\odia a A \in \Gamma|i$ then $init_a : A \in \Gamma|i$ for any $a \in
\Sigma.$
\item If $s : A \in \Gamma|i$ and $s \in F$, then $A \in \Gamma|i.$
\item For all $j$, $a$, $s$ and $t$, such that $i \next{a} j$ and $s
\trans{a}_\Acal t$, 
if $s : A \in \Gamma|i$ then $t : A \in \Gamma|j.$

\item For all $j$, $a$, $s$ and $t$, such that $j \next{a} i$ and $s
\trans{\bar a}_\Acal t$, 
if $s : A \in \Gamma|i$ then $t : A \in \Gamma|j.$
\end{enumerate}
\end{definition}

\begin{definition}[$\Acal$-stability]
A nested sequent $\Gamma$ is {\em $\Acal$-stable} if 
\begin{enumerate}
\item Every node is saturated.
\item $\Gamma$ is $\Acal$-propagated.
\item Every internal node is realised. 
\item For every leaf node $i$, one of the following holds:
  \begin{enumerate}
\item There is an ancestor node $j$ of $i$ such that $\Gamma|i = \Gamma|j.$
We call the node $i$ a {\em loop node}. 
\item $\Gamma|i$ is realised (i.e., it cannot have a member
of the form $\obox a A$).
  \end{enumerate}
\end{enumerate}
\end{definition}

The prove procedure for $\DKm(\Acal)$ is given in Figure~\ref{fig:prove1}. 
We show that the procedure is sound and complete with respect to $\DKm(\Acal)$.
The proofs of the following theorems can be found in the appendix.
\begin{theorem}
\label{thm:auto correctness}
If $Prove_1(\Acal,\{F\})$ returns $\top$ then $F$ is provable in $\DKm(\Acal).$
If $Prove_1(\Acal,\{F\})$ returns $\bot$ then $F$ is not provable in
$\DKm(\Acal).$
\end{theorem}

\begin{theorem}
\label{thm:auto-proc-terminates}
For every nested formula  $A$, $Prove_1(\Acal, \{A\})$ terminates. 
\end{theorem}
\begin{corollary}
The proof system $\DKm(\Acal)$ is decidable.
\end{corollary}

\subsection{A grammar-based procedure}
\label{sec:grammar-proc}

The grammar-based procedure differs from the automaton-based procedure in
the notion of propagation and that of a stable nested sequent. 
In the following, given a function $\theta$ from labels to labels, 
and a list $\vec i = i_1,\dots,i_n$ of labels, we write
$\theta(\vec i)$ to denote the list $\theta(i_1),\dots,\theta(i_n).$
We write $[\vec i \assign \theta(\vec i)]$ to mean the mapping 
$[i_1 \assign \theta(i_1), \ldots, i_n \assign \theta(i_n)].$

In the following definitions, $S$ is assumed to be a context-free
semi-Thue system over some alphabet $\Sigma.$
\begin{definition}[$S$-propagation]
Let $\Gamma$ be a nested sequent. 
Let $\Pcal = (\Sigma,Q,\{i\},\{j\},\delta)$ be a propagation automata,
where $Q$ is a subset of the nodes in $\Gamma.$
We say that $\Gamma$ is {\em $(S,\Pcal)$-propagated}
if the following holds:
$\odia a A \in \Gamma|i$ and $\Pcal \cap L_a(S) \not = \emptyset$
imply $A \in \Gamma|j.$
$\Gamma$ is {\em $S$-propagated} if 
it is $(S, \Rcal(\Gamma,i,j))$-propagated 
for every node $i$ and $j$ in $\Gamma$.
\end{definition}

\begin{definition}[$S$-stability]\label{def:S-stable}
A nested sequent 
$\Gamma$ is {\em $S$-stable} if 
\begin{enumerate}
\item Every node is saturated.
\item $\Gamma$ is $S$-propagated.
\item Every internal node is realised. 
\item 
Let $\vec x = x_1,\dots,x_n$ be the list
of all unrealised leaf nodes.
There is a function $\lambda$ assigning 
each unrealised leaf node $x_m$ to an ancestor $\lambda(x_m)$ of $x_m$
such that $\Gamma|x_m = \Gamma|\lambda(x_m)$ and for every node $y$ and $z$,
$\Gamma$ is $(S,\Pcal)$-propagated, where 
$
\Pcal = \Rcal(\Gamma,y,z)[\vec x \assign \lambda(\vec x)].
$
\end{enumerate}
\end{definition}

\begin{figure}[t]
$Prove_2(S, \Gamma,k)$
\begin{enumerate}
\item If $\Gamma = \Gamma'[p,\neg p]$, return $\top$.
\item If $\Gamma$ is $S$-stable, return $\bot.$
\item If $\Gamma$ is not saturated: 
\begin{itemize}
\item If $A \lor B \in \Gamma|i$ but $A \notin \Gamma|i$
or $B \notin \Gamma|i$, then let $\Gamma' := \Gamma(i \addtree \{A,B\})$
and return $Prove_2(S,\Gamma',k).$

\item Suppose $A_1 \land A_2 \in \Gamma|i$ but neither $A_1 \in \Gamma|i$
nor $A_2 \in \Gamma|i$. Let $\Gamma_1 = \Gamma(i \addtree \{A_1\})$
and $\Gamma_2 = \Gamma(i \addtree \{A_2\}).$
Then return $\bot$ if $Prove_2(S,\Gamma_j,k) = \bot$ for some $j \in \{1,2\}.$
Otherwise return $\star$ if $Prove_2(S,\Gamma_j,k) = \star$ for some $j \in \{1,2\}.$
Otherwise return $\top.$
\end{itemize}

\item If $\Gamma$ is not $S$-propagated: 
then there must be nodes $i$ and $j$ such that
$\odia a A \in \Gamma|i$ and 
$\Rcal(\Gamma,i,j) \cap L_a(S) \not = \emptyset$,
but $A \not \in \Gamma|j.$
Let $\Gamma' := \Gamma(j \addtree \{A\}).$
Return $Prove_2(S, \Gamma',k).$

\item If there is an internal node $i$ in $\Gamma$ that is not realised:
Then there is $\obox a A \in \Gamma|i$ such that $A \not \in \Gamma|j$
for every $j$ s.t. $i \next{a} j.$ 
Let $\Gamma' := \Gamma(i \addtree \seq a {A}).$
Return $Prove_2(S,\Gamma',k).$

\item Non-deterministically choose a leaf node $i$ that is not realised and is
at height equal to or lower than $k$ in $\Gamma$: 
Then there is $\obox a A \in \Gamma|i$.
Let $\Gamma' := \Gamma(i \addtree \seq a {A}).$
Return $Prove_2(S,\Gamma',k).$

\item Return $\star$.
\end{enumerate}
\caption{A grammar-based prove procedure.}
\label{fig:prove2}
\end{figure}

Now we define a non-deterministic prove procedure $Prove_2(S,\Gamma,k)$ as in
Figure~\ref{fig:prove2}, where $k$
is an integer and $S$
is a context-free closed semi-Thue system. Given a nested sequent $\Gamma$, and
a node $i$ in $\Gamma$, the {\em height} of $i$ in $\Gamma$ is the length of the
branch from the root of $\Gamma$ to node $i.$ 
The procedure $Prove_2(S,\Gamma,k)$ tries to construct a derivation of $\Gamma$,
but is limited to exploring only those nested sequents derived from $\Gamma$ that
has height at most $k.$  
The procedure $Prove$ given below is essentially an iterative deepening procedure
that calls $Prove_2$ repeatedly with increasing values of $k$. If an input sequent
is not valid, the procedure will try to guess the smallest $S$-stable sequent
that refutes the input sequent, i.e., it essentially tries to construct a
finite countermodel. 
\begin{enumerate}
\item[] \noindent $Prove(S,\Gamma)$
\item $k:=0$.
\item If $Prove_2(S,\Gamma,k)=\top$ or $Prove_2(S,\Gamma,k)=\bot$, return
  $\top$ or $\bot$ respectively.
\item $k:=k+1$. Go to step (ii).
\end{enumerate}

The procedure $Prove$ gives a semi-decision procedure for context-free 
grammar logics. This uses the following lemma about $S$-stable sequents, which
shows how to extract a countermodel from an $S$-stable sequent. 

\begin{lemma}
\label{lm:S-stable}
Let $S$ be a context-free closed semi-Thue system.  
If $\Gamma$ is an $S$-stable nested sequent, then there exists
a model $\mathfrak{M}$ such that for every node $x$ in $\Gamma$
and for every $A \in \Gamma|x$, there exists a world $w$ in $\mathfrak{M}$
such that $\mathfrak{M}, w \not \models A.$
\end{lemma}

\begin{theorem}
\label{thm:Prove-sound-complete}
Let $S$ be a context-free closed semi-Thue system.  
For every formula $F$, $Prove(S, \{F\})$ returns $\top$ if and only if
$F$ is provable in $\DKm(S).$
\end{theorem}

We next show that $Prove(S,\Gamma)$ terminates when $S$ is regular.
The key is to bound the size of $S$-stable sequents, hence the
non-deterministic iterative deepening will eventually find
an $S$-stable sequent, when $\Gamma$ is not provable. 

\begin{theorem}
\label{thm:grammar-proc-terminates}
Let $S$ be a regular closed semi-Thue system over an alphabet $\Sigma$. 
Then for every formula $F$, the procedure $Prove(S, \{F\})$ terminates. 
\end{theorem}
The proof relies on the fact that there exists a minimal FSA $\Acal$ encoding $S$, so one
can simulate steps of $Prove_1(\Acal,\{F\})$ in $Prove(S,\{F \}).$ 
It is not difficult to show that if a run of $Prove_1(\Acal,\{F\})$ reaches 
a $\Acal$-stable nested sequent $\Gamma'$, then one can find a $k$ such that
a run of $Prove_2(S,\{F\},k)$ reaches a saturated and $S$-propagated nested sequent 
$\Delta$, such that $\Gamma'$ and $\Delta$ are identical except for the
labeled formulae in $\Gamma'$. The interesting part is in showing that $\Delta$ is $S$-stable. 
The details are in the appendix.

The following is then a corollary of Theorem~\ref{thm:Prove-sound-complete} 
and Theorem~\ref{thm:grammar-proc-terminates}.  
\begin{corollary}
Let $S$ be a regular closed semi-Thue system over an alphabet $\Sigma$. 
Then the procedure $Prove$ is a decision procedure for $\DKm(S).$
\end{corollary}

\section{Conclusion and future work}

Nested sequent calculus is closely related to display calculi,
allowing us to benefit from well-studied proof theoretic techniques in
display calculi, such as Belnap's generic cut-elimination procedure,
to prove cut-elimination for $\SKm(S)$.  At the more practical end, we
have established via proof theoretic means that nested sequent
calculi for regular grammar logics can be effectively mechanised.
This work and our previous work~\cite{Gore09tableaux,Gore11LMCS}
suggests that nested sequent calculus could potentially be a good
intermediate framework to study both proof theory and decision
procedures, at least for modal and substructural logics.

Nested sequent calculus can be seen as a special case of labelled sequent calculus,
as a tree structure in a nested sequent can be encoded using labels and accessibility
relations among these labels in labelled calculi. 
The relation between the two has recently been established in \cite{ramanayake11phd}, where the authors
show that, if one gets rid of the frame rules in labelled calculi and structural rules in nested sequent calculi,
there is a direct mapping between derivations of formulae between the two frameworks. 
However, it seems that the key to this connection, i.e., admissibility of the frame rules, 
has already been established in Simpson's thesis~\cite{simpson94phd},\footnote{Simpson's results 
are shown for intuitionistic modal logics, but it is straightforward to apply the techniques shown
there to classical modal logics}
where he shows admissibility of a class of frame rules (specified via Horn clauses) 
in favor of propagation rules obtained by applying a 
closure operation on these frame rules. The latter is similar to our notion of propagation rules.
Thus it seems that structural rules in (shallow) nested sequent calculus play a similar role
to the frame rules in labelled calculi. We plan to investigate this connection
further, e.g., under what conditions the structural rules are admissible in deep inference 
calculi, and whether those conditions translate into any meaningful characterisations in terms
of (first-order) properties of frames. 

The two decision procedures for regular grammar logics we have
presented are not optimal. As can be seen from the termination 
proofs, their complexity is at least EXPSPACE.  
We plan to refine the procedures further to achieve optimal
EXPTIME complexity, e.g, by extending our deep nested 
sequent calculi with ``global caching'' techniques from tableaux
systems~\cite{GoreW09}.

\paragraph{Acknowledgment} The authors would like to thank an anonymous reader
of a previous draft for his/her detailed and useful comments. 
The first author is supported by the Australian Research Council Discovery Grant
DP110103173.

\bibliography{biblio}
\bibliographystyle{aiml12}

\appendix
\newcommand{\thmhead}[2]{\noindent {\bf #1 \ref{#2}.}}

\section{Proofs}

\thmhead{Theorem}{thm:Km-S}
A formula $F$ is $S$-valid iff $F$ is provable in $\Km(S).$
\begin{proof}
The soundness and completeness proofs follow 
the same proofs in \cite{baldoni98phd} for axiomatisations
of grammar logics without converse. The soundness proof
is quite straightforward so we omit them.
For the completeness proof, it is enough to show that
the construction of canonical models in \cite{baldoni98phd}
additionally satisfies the residuation axiom, and the
rest of the proof is the same. 
The canonical models are defined using the notion of maximal consistent sets.
A formula $A$ is said to be consistent if $\neg A$ is not provable
in $\Km(S).$ A finite set of formulae is consistent if the conjuction of all of
them is consistent, and an infinite set is consistent if every finite subset of
it is consistent. A set of formulae $\Scal$ is maximally consistent if it is consistent
and for every formula $A$, either $A \in \Scal$ or $\neg A \in \Scal.$
Following \cite{baldoni98phd}, it can be shown that a maximal consistent set $\Scal$
satisfies, among others, the following:
\begin{itemize}
\item There is no formula $A$ such that $A \in \Scal$ and $\neg A \in \Scal.$
\item If $A \in \Scal$ and $A \impl B \in \Scal$ then $B \in \Scal.$
\item If $A$ is provable in $\Km(S)$ then $A \in \Scal.$
\end{itemize}
We now define the canonical model $\Mf_c = \langle W, \{R_a\}_{a\in \Sigma}, V \rangle$ 
as follows:
\begin{itemize}
\item $W$ is the set of all maximal consistent sets. 
\item For every $a \in \Sigma$, 
$R_a = \{(w,w') \mid w_a \subseteq w' \}$ 
where $w_a = \{A \mid [a]A \in w \}.$
\item For each propositional variable $p$, $V(p) = \{w \mid p \in w\}.$
\end{itemize}
It is enough to show that $R_a = R_{\bar a}^{-1}$, i.e., that $R_a$ is the inverse
of $R_{\bar a}.$ This is proved by contradiction. 

Suppose otherwise, i.e., there exists $w$ and $w'$ such that
$(w,w') \in R_a$ but $(w',w) \not \in R_{\bar a}.$
This means that there exists $[\bar a]A \in w'$ such that $A \not \in w.$
Because $w$ is maximally consistent, we have $\neg A \in w.$
Since we have an instance of the residuation axiom $\neg A \impl \obox{a} \odia{\bar a} \neg A \in w$
and since maximally consistent sets are closed under modus ponens, 
we also have $\obox{a}\odia{\bar a}\neg A \in w.$
Because $(w,w') \in R_a$, the latter implies that $\odia{\bar a} \neg A \in w'.$
But this means $\odia{\bar a}\neg A = \neg (\obox{\bar a} A) \in w'$, contradicting
the consistency of $w'.$
The rest of the proof then proceeds as in \cite{baldoni98phd} (Chapter II).
Briefly, one shows that for every $w$ and $A$, if $A \in w$ then
$\Mf_c, w \models A$. Now if $A$ is $S$-valid but not provable in $\Km(S)$, then $\neg \neg A$
is not provable either. This means $\neg A$ is in some maximal consistent set $w$, 
and therefore $\Mf_c, w\models \neg A$, and $\Mf_c, w \not \models A$,
contradicting the validity of $A.$ 
\qed
\end{proof}

\thmhead{Theorem}{thm:SKm-Km-equiv}
The system $\SKm(S)$ and $\Km(S)$ are equivalent. 
\begin{proof}
{\em (Outline).}
In one direction, from  $\SKm(S)$ to $\Km(S)$, we show that,
for each inference rule of $\SKm(S)$, if the formula interpretation of
the premise(s) is valid then the formula interpretation of the conclusion
is also valid. For the converse, it is enough to show that all axioms
of $\Km(S)$ are derivable in $\SKm(S).$
It can be shown that both the residuation axioms and the axioms 
generated from $S$ can be derived using the structural rules $r$ and
$\rho(S).$ For example, suppose $S$ contains the axiom $[a][b] p \impl [c][d] p.$ Then 
the (nnf of the) axiom can be derived as shown in the figure on the right
(where a double-line indicates one or more application of rules):
$$
\infer[\lor]
{\odia a \odia b \neg p \lor \obox c \obox d p}
{
\infer[\obox c]
{\odia a \odia b \neg p, \obox c \obox d p}
{
 \infer[r]
 {\odia a \odia b \neg p, \seq c {\obox d p}}
 {
  \infer[\obox d]
  {\seq{\bar c}{\odia a \odia b \neg p}, \obox d p}
  {
    \infer[r]
    {\seq{\bar c}{\odia a \odia b \neg p}, \seq{d} p}
    {
      \infer[\rho(S)]
      {\seq{\bar d}{\seq{\bar c}{\odia a\odia b \neg p}}, p}
      {
       \infer=[r]
       {\seq{\bar b}{\seq{\bar a}{\odia a\odia b\neg p}}, p}
       {
         \infer[\odia a]
         {\odia a\odia b \neg p, \seq a {\seq b p}}
         {
          \infer[r]
          { \seq a {\odia b \neg p, \seq b p}}
          {
            \infer[\odia b]
            {\seq {\bar a} {~}, \odia b \neg p, \seq b p}
            {
              \infer[r]
              {\seq a {~}, \seq b {\neg p, p}}
              {
                \infer[id]
               {\seq {\bar b} {\seq{\bar a}{~}}, \neg p, p} {}
              }
            }
          }
         }
       }
      }
    }
  }
 }
}
}
$$
\end{proof}

\thmhead{Lemma}{lm:adm-r}
The rule $r$ is height-preserving admissible in $\DKm(S).$
\begin{proof}
Suppose $\Pi$ is a derivation of $\Gamma, \seq a \Delta$.
We show by induction on $|\Pi|$ that there exists a derivation
$\Pi'$ of $\seq {\bar a} \Gamma, \Delta$ such that $|\Pi| = |\Pi'|.$
This is mostly straightforward, except for the case where
$\Pi$ ends with a propagation rule. In this case, it is enough
to show that the propagation automata for $\Gamma, \seq a \Delta$
is in fact exactly the same as the propagation automata of $\seq {\bar a} \Gamma, \Delta$.
\end{proof}

\thmhead{Lemma}{lm:adm-medial}
The rules $\mathit{actr}$ and $m$ are height-preserving admissible. 
\begin{proof}
Admissibility of $\mathit{actr}$ is trivial.
To show admissibility of $m$, the non-trivial case is when we need to permute
$m$ over $p_S.$
Suppose $\Pi$ is a derivation of $\Gamma[\seq a {\Delta_1}, \seq a {\Delta_2}]$
ending with a propagation rule. 
Suppose $i$ is the node where $\Delta_1$ is located and $j$ is the node
where $\Delta_2$ is located. 
If $\Pcal$ is a propagation automata between nodes $k$ and $l$ in $\Gamma[\seq a {\Delta_1}, \seq a {\Delta_2}]$,
then $\Pcal[j \assign i]$ is a propagation automata between nodes $k[j \assign i]$ and $l[j \assign i]$ in 
$\Gamma[\seq a {\Delta_1,\Delta_2}]$. So all potential propagations of diamond formulae
are preserved in the conclusion of $m.$ So $m$ can be permuted up over the propagation rule
and by the induction hypothesis it can be eventually eliminated.
\end{proof}

\thmhead{Lemma}{lm:adm-rho-S}
The structural rules $\rho(S)$ of $\SKm(S)$ are height-preserving 
admissible in $\DKm(S).$
\begin{proof}
Suppose $\Pi$ is a derivation of $\Gamma[\seq a \Delta]$.
We show that there is a derivation $\Pi'$ of
$\Gamma[\seq {u} \Delta]$, where $u = a_1\cdots a_n$
such that $a \prod u \in S.$
This is mostly straightforward except when $\Pi$ ends with
a propagation rule. 
Suppose the hole in $\Gamma[~]$ is located at node $k$ and
$\Delta$ is located at node $l$, with $k \next{a} l.$
In this case we need to show that
if a diamond formula $\odia b A$ can be propagated from a node $i$ to node $j$
in $\Gamma[\seq {a} \Delta]$ then there is also a propagation
path between $i$ and $j$ in $\Gamma[\seq {u} \Delta]$ for the same formula.
Suppose $\Pcal_1$ is the propagation automata $\Rcal(\Gamma[\seq {a} \Delta],i,j).$
Then the propagation automata $\Pcal_2 = \Rcal(\Gamma[\seq u \Delta],i,j)$
is obtained from $\Pcal_1$ by adding $n-1$ new states $k_1,\dots,k_{n-1}$
between $k$ and $l$, and the following transitions: 
$k \trans{a_1} k_1$, $k_1 \trans{a_{m+1}} k_{m+1}$, for $2 \leq m < n$
and $k_{n-1} \trans{a_n} l$, and their dual transitions.

Suppose $i \trans{v} j$ is a propagation path in $\Gamma[\seq a \Delta].$
If 
$v$ does not go through the edge $k \next{a} l$ (in either direction, up or down) 
then the same path also exists in $\Gamma[\seq {u} \Delta].$
If it does pass through $k \next{a} l$, then the path 
must contain one or more transitions of the form $k \trans{a} l$ or $l \trans{\bar a} k$.
Then one can simulate the path $i \trans{v} j$ with a path $i \trans{v'} j$ in $\Pcal_2$, where
$v'$ is obtained from $v$ by replacing each $k \trans{a} l$ with $k \trans{u} l$ and
each $l \trans{\bar a} k$ with $l \trans{\bar u} k.$
It remains to show that $v' \in \Pcal_2 \cap L_b(S).$ But this follows
from the fact that $a \prod u \in S$ and $\bar a \prod \bar u \in S$ (because $S$ is a
closed), so $v \prods_S v' \in L_b(S).$
\end{proof}

\thmhead{Theorem}{thm:SKm-DKm-equiv}
For every context-free closed semi-Thue system $S$, 
the proof systems $\SKm(S)$ and $\DKm(S)$ are equivalent.
\begin{proof}
One direction, from $\SKm(S)$ to $\DKm(S)$ follows from the admissibility of
structural rules of $\SKm(S)$ in $\DKm(S).$ 
To show the other direction, given a derivation $\Pi$ in $\DKm(S)$,  we show, by induction on the number
of occurrences of $p_S$, with a subinduction on the height of $\Pi$, that $\Pi$ can be transformed
into a derivation in $\SKm(S).$ As rules other than $p_S$ can be derived directly in $\SKm(S)$, the
only interesting case to consider is when $\Pi$ ends with $p_S$: 
$$
\infer[p_S, \mbox{ where $\Rcal(\Gamma[~]_i[~]_j, i, j) \cap L_a(S) \not = \emptyset$}]
{\Gamma[\odia a A]_i[\emptyset]_j}
{\Gamma[\odia a A]_i[A]_j}
$$ 
and $\Gamma[\odia a A]_i[A]_j$ is derivable via a derivation $\Pi'$ in $\DKm(S).$
Choose some $u \in \Rcal(\Gamma[~]_i[~]_j, i, j) \cap L_a(S).$ 
Then we can derive the implication $\odia u A \impl \odia a A$ in $\SKm(S)$. 
Using this implication, the display property and the cut rule, it can be shown that 
the following rule is derivable in $\SKm(S).$
$$
\infer[d]
{\Gamma[\odia a A]}
{\Gamma[\odia a A, \odia u A]}
$$
Then we show that  
the rule $p_S$ can be simulated by the derived rule $d$ above, with
chains of $\odia a$-rules in $\SKm(S)$, and utilising the weakening lemma
(Lemma~\ref{lm:weak}). 

Suppose $u = a_1 \cdots a_n$. Then there are nodes $s_1,\dots,s_n$ in $\Gamma[]_i[]_j$, 
with $s_1 = i$ and $s_n=j$, such that the following is a path in the
propagation automaton $\Rcal(\Gamma[~]_i[~]_j,i,j)$:
$$
i = s_1 \trans{a_1} s_2 \trans{a_2} \cdots s_{n-1} \trans{a_n} s_n = j
$$
Now instead of propagating $A$ using $p_S$ applied to $\odia a A$, we can propagate $A$ in stages 
using $\odia u A$ and the diamond rules $\odia {a_1}, \ldots, \odia {a_n}.$
Let $\Gamma'[]_i[]_j$ be a context obtained from $\Gamma[]_i[]_j$ by 
adding the formula $\odia {a_1} \cdots \odia{a_{n-k + 1}} A$ to node $s_k$, for each $1 \leq k \leq n.$
Then it can be shown, by induction on $n$, that we have a derivation
$$
\deduce{\Gamma[\odia a A, \odia u A]_i[\emptyset]_j}
{\deduce{\vdots}{\Gamma'[\odia a A, \odia u A]_i[A]_j}}
$$
in $\DKm(S)$ using only the diamond rules $\odia {a_1}, \dots, \odia{a_n}.$ Note that
as these are diamond rules, not $p_S$, they can be simulated in $\SKm(S)$, so the above
derivation can be simulated as well in $\SKm(S).$
By the weakening lemma (Lemma~\ref{lm:weak}), we can construct a derivation
$\Psi$ of $\Gamma'[\odia a A, \odia u A]_i[A]_j$, such that the height of $\Pi'$
is the same as $\Psi.$ 
So by the induction hypothesis we have a derivation $\Psi'$ of 
$\Gamma'[\odia a A, \odia u A]_i[A]_j$ in $\SKm(S).$ The final derivation in $\SKm(S)$ is thus
constructed by chaining the above derivations:
$$
\infer[d]
{\Gamma[\odia A]_i[\emptyset]_j}
{
\deduce{\Gamma[\odia a A, \odia u A]_i[\emptyset]_j}
{\deduce{\vdots}{\deduce{\Gamma'[\odia a A, \odia u A]_i[A]_j}{\Psi'}}} 
}
$$
\end{proof}

\thmhead{Theorem}{thm:DKm-A-eq-S}
Let $S$ be a regular closed semi-Thue system over $\Sigma$ and let $\Acal$ be a 
FSA representing the regular language generated by $S$ and $\Sigma.$
Then $\DKm(S)$ and $\DKm(\Acal)$ are equivalent.
\begin{proof}
{\em (Outline).} 
To show that if a formula $B$ is provable in $\DKm(S)$ then $B$ is provable in
$\DKm(\Acal)$ we will demonstrate that given a proof of $B$ in
$\DKm(S)$ it is possible to replace the highest application of a propagation rule
from $\DKm(S)$ with a sequence of propagation rules from $\DKm(\Acal)$. As
all non-propagation rules between the two systems are identical, this will be
sufficient to show that a proof in $\DKm(S)$ can be translated to a proof in
$\DKm(\Acal)$.

Suppose we have a derivation $\Pi$ of $\Gamma[\odia a A]_i[A]_j$ using only the
rules of $\DKm$. If $p_S$ is applicable and yields $\Gamma[\odia a
A]_i[\emptyset]_j$, it must be the case that $\Rcal(\Gamma[]_i[]_j,i,j)\cap
L_a(S)\neq\emptyset$. Therefore there exists a sequence of transitions
in $\Rcal(\Gamma[~]_i[~]_j,i,j)$: 
$
i \trans{a_1} i_1 \trans{a_2} \cdots \trans{a_{n-1}} i_{n-1} \trans{a_n} j, 
$
where $a_1\cdots a_n \in L_a(S)$ and 
where each $i_k$, for $1 \leq k \leq n-1$, is a node in $\Gamma[~]_i[~]_j$
and 
\begin{itemize}
\item either $i \next{a_1} i_1$ or $i_1 \next{\bar a_1} i$, 
\item either $i_{k-1} \next{a_{k}} i_{k}$ or $i_{k} \next{\bar a_k} i_{k-1}$, for $2 \leq k < n-1$
\item and either $i_{n-1} \next{a_n} j$ or $j \next{\bar a_n} i_{n-1}.$
\end{itemize}
Since $\Acal$ accepts $L_a(S)$, there must exist a sequence of transitions
in $\Acal$ such that: 
$
init_a \trans{a_1} s_1 \trans{a_2} \cdots \trans{a_{n-1}} s_{n-1} \trans{a_n} f, 
$
where $f$ is a final state in $\Acal.$
The propagation path $a_1\cdots a_n$ can then be simulated in $\DKm(\Acal)$ as follows.
First, define a sequence of nested sequents as follows:
\begin{itemize}
\item $\Gamma_0 := \Gamma[\odia a A]_i[\emptyset]_j$, $\Gamma_1 := \Gamma[\odia a A, init_a : A]_i[\emptyset]_j.$
\item $\Gamma_{k+1} := \Gamma_k(i_k \addtree \{s_k : A\})$, for $1 \leq k \leq n-1$.
\item $\Gamma_{n+1} := \Gamma_{n}(j \addtree \{f : A\})$ and $\Gamma_{n+2} := \Gamma_{n+1}(j \addtree \{A\}).$
\end{itemize}
Then $\Gamma_0$ can be obtained from $\Gamma_{n+2}$ by a series of applications of
propagation rules of $\DKm(\Acal)$. That is, $\Gamma_0$ is obtained from $\Gamma_1$
by applying the rule $i$; $\Gamma_k$ is obtained from $\Gamma_{k+1}$ by applying
either the rule $\tdn$ or $\tup$, for $1 \leq k \leq n-1$, at node $i_{k}$ and
$\Gamma_{n}$ is obtained from $\Gamma_{n+1}$ by applying the rule $\tdn$ or $\tup$ at node $j$,
and $\Gamma_{n+1}$ is obtained from $\Gamma_{n+2}$ by applying the rule $f$ at node $j.$
Note that $\Gamma_{n+2}$ is a weakening of $\Gamma[\odia a A]_i[A]_j$ with labeled
formulae spread in some nodes between $i$ and $j.$
It remains to show that $\Gamma_{n+2}$ is derivable. This is obtained simply by
applying weakening (Lemma~\ref{lm:weak}) to $\Pi.$ 

For the other direction, assume we have a $\DKm(\Acal)$-derivation $\Psi$ of $B$. 
We show how to construct a derivation $\Psi'$ of $B$ in $\DKm(S).$
The derivation $\Psi'$ is constructed as follows:
First, remove all labelled formulae from $\Psi$; then remove the rules
$\tup$, $\tdn$ and $i$, and finally, replace the rule $f$ with $p_S.$
The rules $\tup$, $\tdn$ and $i$ from $\Psi$ simply disappear in $\Psi'$ because with 
labelled formulae removed, the premise and the conclusion of any of the rules
in $\Psi$ map to the same sequent in $\Psi'.$ Instances of the other rules in $\Psi$
map to the same rules in $\Psi'.$
We need to show that $\Psi'$ is indeed a derivation in $\DKm(S).$ The only non-trivial
case is to show that the mapping from the rule $f$  to the rule $p_S$ is correct, i.e.,
the resulting instances of $p_S$ in $\Psi'$ are indeed valid instances.

We first prove an invariant property that holds for $\Psi$.
We say that a nested sequent $\Delta$ is {\em $\Acal$-connected} iff the following hold
\begin{itemize}
\item If $init_a : C \in \Delta | i$ then $\odia a C \in \Delta | i.$
\item If $s : C \in \Delta | i$ and $s$ is not an initial state of $\Acal$, then
there exists an $a \in \Sigma$ and a sequence of nodes 
$x_1,\dots, x_n$ in $\Delta$ and a sequence of states $s_1,\dots,s_n$ of $\Acal$ such that
\begin{itemize}
\item $s_k : C \in \Delta | x_k$ for $1 \leq k \leq n.$
\item $s_1 = init_a$ and $x_n = i$. 
\item For each $1 \leq k < n$, 
$s_k \trans{b}_\Acal s_{k+1}$ for some $b \in \Sigma$, and
either $x_k \next{b} x_{k+1}$ or $x_{k+1} \next{\bar b} x_k.$
\end{itemize}
\end{itemize}
It is then easy to verify the following claim:
\begin{quote}{\bf Claim:}
If $\Delta$ is $\Acal$-connected and there is a derivation $\Xi$ 
of $\Delta$ in $\DKm(\Acal)$, then every nested sequent in $\Xi$ is
$\Acal$-connected. 
\end{quote}
Given the above claim, and the fact that the nested sequent $\{B\}$ is
trivially $\Acal$-connected, it follows that every nested sequent in $\Psi$ is
$\Acal$-connected. 
Now, it remains to show each instance of $f$ in $\Psi$ 
can be replaced by a valid instance of $p_S$ in $\Psi'.$
Suppose there is an instance of $f$ in $\Psi$ as shown below left: 
$$\infer[f]{\Gamma[s:A]_j}{\Gamma[s:A,A]_j}
\qquad
\infer[p_S]
{\Gamma''[\odia a A]_i[\emptyset]_j}{\Gamma''[\odia a A]_i[A]_j}
$$
Then we by the above claim, there must exist a node $i$ and an $a \in \Sigma$ 
such that $\odia a A \in \Gamma[s:A]_j | i$ and 
that there exist a sequence of nodes $i = x_1,\dots, x_n = j$
and a sequence of states $init_a = s_1, \dots, s_n = s$
such that $s_1 \trans{a_1}_\Acal s_2 \trans{a_2}_\Acal \cdots \trans{a_{n-1}}_\Acal s_n$
for some $a_1,\dots,a_{n-1}.$ It also follows from $\Acal$-connectedness
that $a_1\cdots a_{n-1}$ is an element of the propagation automata
$R(\Gamma[s:A]_j, i, j).$
Because $\Acal$ represents the regular languages $\{L_b(S) \mid b \in \Sigma \}$, we have
that $a_1 \cdots a_{n-1} \in L_a(S)$, and 
\begin{equation}
\label{eq:DKm-A-DKm-S}
a_1 \cdots a_{n-1} \in L_a(S) \cap R(\Gamma[s:A]_j, i,j).
\end{equation}
Let $\Gamma'[\odia a A]_i[s:A]_j = \Gamma[s:A]_j.$ 
Let $\Gamma''[~]_i[~]_j$ be the context obtained from $\Gamma'[~]_i[~]_j$
by removing all labelled formulae. Then (\ref{eq:DKm-A-DKm-S}) can
be rewritten as: 
$$a_1 \cdots a_{n-1} \in L_a(S) \cap R(\Gamma''[\odia a A]_i[\emptyset]_j, i, j).$$
Thus the propagation instance $p_S$ shown above right, to which the above instance of $f$ maps to,
is indeed a valid instance of $p_S.$
\end{proof}

\thmhead{Theorem}{thm:auto correctness}
If $Prove_1(\Acal,\{F\}) = \top$ then $F$ is provable in $\DKm(\Acal).$
If $Prove_1(\Acal,\{F\}) = \bot$ then $F$ is not provable in
$\DKm(\Acal).$
\begin{proof}
The proof of the first statement is straightforward, since the steps of $Prove_1$ are just
backward applications of rules of $\DKm(\Acal).$ 
To prove the second statement, we show that if 
$Prove_1(\Acal,\{F\})=\bot$ then there exists a model
$\mathfrak{M}=(W, R,V)$, where $R = \{R_a\}_{a\in\Sigma}$,
such that $\mathfrak{M}\ \slashed{\models}\ F.$ By the completeness of
$\DKm(\Acal)$, it will follow that $F$ is not provable in $\DKm(\Acal)$.

Since $Prove_1(\Acal,\{F\})=\bot$ the procedure must generate
an $\Acal$-stable $\Delta$, with $F$ in the root node of $\Delta.$ 
Let $W$ be the set of all the realised nodes of $\Delta$. For every pair 
$i,j \in W$, construct an automaton $\Pcal(i,j)$ by
modifying the propagation automaton $\Rcal(\Delta,i,j)$ by identifying every
unrealised node $k'$ with its closest ancestor $k$ such that
$\Delta|k=\Delta|k'$. That is, replace every transition of the form
$s\trans{a} k'$ with $s\trans{a} k$ and $k'\trans{a} s$
with $k\trans{a} s$. Then define $R_a(x,y)$ iff
$\Pcal(x,y)\cap L(\Acal_a)\neq\emptyset$, where $\Acal_a$ is
$\Acal$ with only $init_a$ as the initial state.

Suppose $S$ is a closed semi-Thue system that corresponds to $\Acal.$
Then $L(\Acal_b) = L_b(S)$ for every $b \in \Sigma.$
We first show that the $\Sigma$-frame $\langle W, R\rangle$ defined above
satisfies all the production rules in $S$ (see Definition~\ref{def:S-frame}).
Let $a \prod u \in S$, where $u = a_1 \cdots a_n.$
We need to show that $R_u \subseteq R_a.$
Suppose otherwise, that is, there is a sequence of worlds 
$x_1,\dots,x_{n+1}$ such that $x_i R_{a_i} x_{i+1}$ but
$(x_1,x_{n+1}) \not \in R_a.$ 
By the above construction, we have $R_b(x,y)$ iff $\Pcal(x,y) \cap L_b(S) \neq \emptyset$
for every $b \in \Sigma.$ 
So it follows that, for each pair $(x_i, x_{i+1})$,
there is a string $u_i \in \Rcal(x_i,x_{i+1}) \cap L_{a_i}(S).$
It also follows that we have a sequence of transitions 
$x_1 \trans{u_1\cdots u_n} x_{n+1}$ in $\Pcal(x_1,x_{n+1})$, 
by chaining the transitions $x_i \trans{u_i} x_{i+1}$ together. 
Because $a \prod u \in S$, and each $u_i \in L_{a_i}(S)$, we have
$$a \prod a_1\cdots a_n \prods u_1 \cdots u_n \in L_a(S).$$
So $u_1 \cdots u_{n}$ is in $\Rcal(\Delta,x_1,x_{n+1}) \cap L_a(S)$, and
therefore $(x_1,x_{n+1}) \in R_a$, contradicting the assumption. 

To complete the model, let $x \in V(p)$ iff $\neg p \in\Delta|x$. We claim that
for every $x \in W$ and every $A \in \Delta | x$, we have 
$\Mf, x \ \slashed{\models}\ A$. We shall prove this by induction on the
size of $A$. Note that we ignore the labelled formulae in $\Delta$; they
are just a bookeeping mechanism. 
As $F$ is in the root node of $\Delta$, 
this will also prove $\Mf\ \slashed{\models} F.$
We show here the interesting case involving the diamond operators.

Suppose $\odia a  A\in\Delta|x$. Assume for a
contradiction that $\mathfrak{M},x\models\odia a  A$. That is, $R_a(x,y)$
and $\mathfrak{M},y\models A.$ If $R_a(x,y)$ then
there is a accepting path $p_a(x,y)$ in $\Pcal(x,y)$ of the form:
$
x_0 \trans{a_1} x_1 \trans{a_2} x_2 \cdots x_{n-1} \trans{a_n} x_n,
$
where $x_0 = x$ and $x_n = y$
such that $u = a_1 \dots a_n \in L(\Acal_a).$
Then because $u \in L(\Acal_a)$, there must be a sequence of states 
$s_0,s_1,\dots,s_n$ of $\Acal$ such that $s_0 = init_a \in I$ and $s_n\in F$ and the transitions between states
$$
s_0 \trans{a_1} s_1 \trans{a_2} s_2 \cdots s_{n-1} \trans{a_n} s_n.
$$
We show by induction on the length of transtions that that $s_i : A \in \Delta|x_i$
for $0 \leq i \leq n.$
In the base case, because $\odia a A \in \Delta|x$, by $\Acal$-propagation,
we have $s_0 : A \in \Delta|x_0.$
For the inductive cases, suppose $s_i : A \in \Delta|x_i$, for $n > i \geq 0$. There are two cases to consider.
Suppose the transition $x_i \trans{a_{i+1}}_{\Pcal(x,y)} x_{i+1}$ is present in $\Rcal(\Delta,x,y).$
Then either $x_i \next{a_{i+1}} x_{i+1}$ or $x_{i+1} \next{\bar a_{i+1}} x_i.$
In either case, by $\Acal$-propagation of $\Delta$, we must have $s_{i+1} : A \in \Delta|x_{i+1}.$

If $x_i \trans{a_{i+1}}_{\Pcal(x,y)} x_{i+1}$ is not a transition in $\Rcal(\Delta,x,y),$
then this transition must have resulted from a use of a loop node. There are two subcases:
either $x_i$ or $x_{i+1}$ is the closest ancestor of a loop node $x'$ with $\Delta | x_i = \Delta | x'$
or, respectively, $\Delta | x_{i+1} = \Delta | x'.$
Suppose $x_i$ is the closest ancestor of $x'$ with $\Delta | x_i = \Delta | x'.$
By the definition of $\Pcal(x,y)$, 
this means we have $x' \trans{a_{i+1}} x_{i+1}$ in $\Rcal(\Delta,x,y).$
Because $\Delta | x_i = \Delta | x'$ and $s_i : A \in \Delta | x_i$, 
we have $s_i : A \in \Delta|x'.$ Then by $\Acal$-propagation, it must be the case that
$s_{i+1} : A \in \Delta|x_{i+1}.$
Suppose $x_{i+1}$ is the closest ancestor of $x'$ with $\Delta|x' = \Delta|x_{i+1}.$
Then $x_i \trans{a_{i+1}} x'$ is a transition in $\Rcal(\Delta,x,y).$
By $\Acal$-propagation, it must be the case that $s_{i+1}:A \in \Delta|x'$,
and therefore also $s_{i+1} : A \in \Delta|x_{i+1}.$

So we have $s_n : A \in \Delta|y.$ But, again by $\Acal$-propagation,
this means $A \in \Delta|y$ (because $s_n$ is a final state). 
Then by the induction hypothesis, we have $\Mf, y \slashed{\models} A$, contradicting
the assumption. 
\end{proof}

\thmhead{Theorem}{thm:auto-proc-terminates}
For every nested formula  $A$, $Prove_1(\Acal, \{A\})$ terminates. 
\begin{proof}
{\em (Outline)} 
We say that a nested sequent $\Gamma$ is a {\em set-based nested sequent} if
in every node of $\Gamma$, every (labelled) formula occurs at most once (a formula $C$
and its labelled versions are considered distinct). 
By inspection of the procedure $Prove_1$, it is clear that all the intermediate
sequents created during proof search for $Prove_1(\Acal,\{A\})$ 
are set-based sequents.

Steps (i) -- (iv) of the procedure only add (strict) subformulae of formulae occurring in
the input sequent without creating new nodes, so for a given input nested sequent, 
applications of these steps eventually terminate.
Because of the blocking conditions in each step, the same formula cannot be added twice
to a node, so the upper bound of the size of a node (i.e., the number of formulae in it)
is the cardinality of the set of all subformulae in the input sequent, plus all their possible
labellings (which is finite because $\Acal$ has only a finite number of states). 

Step (v) is applicable only to internal nodes which are not realised. So the expansion of
the nested sequent tree in this case adds to the width of the tree, not the height.
It is easy to see that the number of branches in an internal node is bounded by the number
of distinct `boxed' subformulae in the original sequent, so this expansion step cannot be applied indefinitely
without applying step (vi), as the number of distinct boxed subformulae is bounded and no new
internal nodes are created. So the combination of steps (i) -- (v) always terminates for
a given input sequent. The only possible cause of termination is if step (vi) can be applied
infinitely often. We next show that this is not the case. 

The expansion in step (vi) adds to the height of the input nested sequent tree. Because of the
loop checking condition in the step, the height of the trees generated during proof search 
is bounded; we give a more precise bound next. 
Let $m$ be the number of states in $\Acal$ and let $n$ be the number of 
subformulae of $A$. Then the total number of different sets of
formulae and labeled formulae (with labels from $\Acal$) is bounded
by $2^{(m+1)n}.$ Therefore, any set-based nested sequent generated during proof search 
will not cross this bound without creating a loop node. As the height of the trees generated during
proof search is bounded, and the number of branches at each node of the trees is also bounded, 
there are only finitely many possible nested sequent trees that can be generated in each branch of the proof
search. Note that every recursive call in the proof procedure adds something to the input nested sequent,
so every branch in the proof search generates pairwise distinct (set-based) nested sequents. As the number of
possible set-based nested sequents is bounded, the depth of the search is bounded, and because the branching
in proof search is also bounded (i.e., it is a binary branch, created when applying the $\land_d$ rule in step (iii)), 
the search tree must be finite, and thefore the search procedure must terminate.
\end{proof}

\thmhead{Lemma}{lm:S-stable}
Let $S$ be a context-free closed semi-Thue system.  
If $\Gamma$ is an $S$-stable nested sequent, then there exists
a model $\mathfrak{M}$ such that for every node $x$ in $\Gamma$
and for every $A \in \Gamma|x$, there exists a world $w$ in $\mathfrak{M}$
such that $\mathfrak{M}, w \not \models A.$
\begin{proof}
Let $\vec x = x_1,\dots, x_n$
be the list of (pairwise distinct) unrealised leaf nodes in $\Gamma.$
Because $\Gamma$ is $S$-stable, we have a function $\lambda$ assigning 
each unrealised leaf node $x_i$ to an ancestor node $\lambda(x_i)$ such
that $\Gamma|x_i = \Gamma | \lambda(x_i)$, and for every node $y$ and $z$ in 
$\Gamma$, we have that $\Gamma$ is $(S,\Pcal(y,z))$-propagated, where
$
\Pcal(y,z) = \Rcal(\Gamma,y,z)[\vec x \assign \lambda(\vec x)].
$
Then define 
$
\Mf = \langle W, \{R_a \mid a \in \Sigma\}, V\rangle
$
where 
\begin{itemize}
\item $W$ is the set of nodes of $\Gamma$ minus the nodes $\vec x,$
\item for every $x,y \in W$, $R_a(x,y)$ iff $\Pcal(x,y) \cap L_a(S) \not = \emptyset$, and 
\item $V(p) = \{ x  \in W \mid \neg p \in \Gamma | x \}.$
\end{itemize}
We now show that if $A \in \Gamma|v$ then there is a $w \in W$ 
such that $\Mf, w \not \models A,$ where  
the world $w$ is determined as follows: 
if $v$ is in $\vec x$, then $w = \lambda(v)$; otherwise,
$w = v.$
We prove this by induction on the size of $A.$ 
The only interesting cases are 
those where $A = \odia a C$ or $A = \obox a C$ for some $a$ and $C.$
\begin{itemize}
\item Suppose $A = \odia a C.$ 
Suppose, for a contradiction, that $\Mf, w \models \odia a C.$
That means there exists a $w'$ such that $R_a(w,w')$ and 
$\Mf, w' \models C.$
By the definition of $R_a$, we have that 
$\Pcal(w,w') \cap L_a(S) \not = \emptyset.$
Because $\Gamma$ is $S$-stable, by Definition~\ref{def:S-stable}(iv), 
it is $(S,\Pcal(w,w'))$-propagated. This means that $C \in \Gamma|w'.$
Then by the induction hypothesis, $\Mf, w' \not \models C,$ which
contradicts our assumption.
\item Suppose $A = \obox a C.$
To show $\Mf, w \not \models \obox a C$, it is enough to show
there exists $w'$ such that $R_a(w,w')$ and $\Mf, w' \not \models C.$

Note that $w$ must be an internal node in $\Gamma$, so by the $S$-stability of $\Gamma$, 
node $w$ in $\Gamma$ must be realised. 
Therefore there exists a node $z$ such that $w \next{a} z$ in $\Gamma$
and $C \in \Gamma | z.$
If $z \not \in \vec x$, then let $w' = z$; otherwise, let $w' = \lambda(z).$
In either case, $\Gamma | z = \Gamma | w'$, so in particular, $C \in \Gamma|w'.$
Also, in either case, the propagation automata $\Pcal(w, w')$ contains 
a transition $w \trans{a}_{\Pcal(w,w')} w'$ (in the case where $z \in \vec x$, this is
because $\lambda(z)$ is identified with $z$). 
Obviously, $a \in L_a(S)$, so  
$L_a(S) \cap \Pcal(w,w') \not = \emptyset,$ so by the definition of $R_a$,
we have $R_a(w,w').$
Since $C \in \Gamma|w'$, by the induction hypothesis, $\Mf, w' \not \models C.$
So we have $R_a(w,w')$, and $\Mf, w' \not \models C$, therefore 
$\Mf, w \not \models \obox a C.$
\end{itemize}
\end{proof}

\thmhead{Theorem}{thm:Prove-sound-complete}
Let $S$ be a context-free closed semi-Thue system.  
For every formula $F$, $Prove(S, \{F\})$ returns $\top$ if and only if
$F$ is provable in $\DKm(S).$
\begin{proof}
{\em (Outline)} 
One direction, i.e., $Prove(S,\{F\}) = \top$ implies that $F$ is provable
in $\DKm(S)$, follows from the fact that steps of $Prove$ are simply backward
applications of rules of $\DKm(S).$ 
To prove the other direction, we note that if $F$ has a derivation in $\DKm(S)$,
it has a derivation of a minimal length, say $\Pi$. 
In particular, in such an derivation,
there are no two identical nested sequents in any branch of the derivation.
Because in $\DKm(S)$ each backward application of a  
rule retains the principal formula of the rule, every application of a rule 
in $\Pi$ will eventually be covered by one of the steps of $Prove.$ 
Since there are only finitely many rule applications in $\Pi$, eventually
these will all be covered by $Prove$ and therefore it will terminate. 
For example, if $\Pi$ ends with a diamond (propagation) rule applied to a non-saturated sequent,
the $Prove$ procedure will choose to first saturate the sequent before applying
the propagation rule. Since all rules are invertible, we do not lose any
provability of the original sequent, but the $Prove$ procedure may end up
doing more steps. We need to show, additionally, that  
every sequent arising from the execution of $Prove(S,\{F\})$ is not $S$-stable. Suppose
otherwise, i.e., the procedure produces an $S$-stable sequent $\Delta$. Now it must
be the case that $F$ is in the root node of $\Delta.$ By Lemma~\ref{lm:S-stable}, this means
there exists a countermodel that falsifies $F$, contrary to the validity of $F$. 
\end{proof}

\thmhead{Theorem}{thm:grammar-proc-terminates}
Let $S$ be a regular closed semi-Thue system. 
Then for every formula $F$, the procedure $Prove(S, \{F\})$ terminates. 
\begin{proof}
Since $S$ is regular, there exists an automaton $\Acal$ such that
$Prove_1(\Acal,\{F\})$ terminates. We choose the minimal deterministic finite state automaton
$\Acal$ that corresponds to $S.$

Suppose $Prove_1(\Acal,\{F\})=\top$. Then $F$ must be derivable in $\DKm(\Acal)$ by
Theorem~\ref{thm:auto correctness}. Since $\DKm(\Acal)$ and $\DKm(S)$ are equivalent (Theorem~\ref{thm:DKm-A-eq-S}),
there must also be a derivation of $F$ in $\DKm(S)$. Then by Theorem~\ref{thm:Prove-sound-complete},
$Prove(S,\{F\})$ must terminate and return $\top.$

Suppose $Prove_1(\Acal,\Gamma)=\bot$. Then there exists an $\Acal$-stable
$\Gamma'$ that can be constructed from $\Gamma$ in the execution of
$Prove_1(\Acal,\Gamma)$. It can be shown that a $\Delta$
that is identical to $\Gamma'$ without any labelled formulae can be constructed
in the execution of $Prove_2(S,\Gamma,d)$ for some $d$. 
We claim that $\Delta$ is $S$-stable.
Saturation, propagation and the realisation of internal nodes follow immediately
from the construction, it remains to find a function $\lambda$ as in
Definition~\ref{def:S-stable}. We claim that such a function is given by
$\lambda(x)=y$ where $y$ is the closest ancestor of $x$ in $\Gamma'$ such that 
$\Gamma'|x=\Gamma'|y$. That is, we identify each unrealised leaf
with the same node it would have been identified with in
$Prove_1(\Acal,\Gamma)$. 

Let $\vec i = i_1,\dots,i_l$ be the list of all unrealised leaf nodes in $\Delta$
and let $\Pcal(x,y) = \Rcal(\Delta, x, y)[\vec i \assign \lambda(\vec i)].$
(Note that as the tree structures of $\Gamma'$ and $\Delta$
are identical, we also have $\Pcal(x,y) = \Rcal(\Gamma', x, y)[\vec i \assign \lambda(\vec i)].$)
For a contradiction, suppose there exists $j$ and $k$ such that $\Delta$
is not $(S,\Pcal(j,k))$-propagated, i.e., there exist $\odia a A \in \Delta|j$, 
such that  $A\notin\Delta|k$ but $\Pcal(j,k) \cap L_a(S)\neq\emptyset.$ In other words, there is a word
$b_1\dots b_n\in \Pcal(j,k) \cap L_a(S)$, 
and a sequence of states $x_0,\dots,x_n$ in $\Pcal(j,k)$ such that
$x_0=j,x_n=k,x_{m-1}\trans{b_m}_{\Pcal(j,k)} x_m,$ where $1 \leq m < n.$
We will show that there exists a function $St$ assigning states of $\Acal$ to
nodes of $\Gamma'$ satisfying:
$St(x_0)\in I$,
$St(x_{m-1})\trans{b_m}_{\Acal} St(x_m)$,
$St(x_n)\in F$,
and 
$St(x_m):A\in\Gamma'|x_m.$
This will establish that $St(x_n):A\in\Gamma'|x_n$ where $St(x_n)\in
F$. Then by $\Acal$-propagation, it will follow that $A\in\Gamma'|k$,
and therefore $A \in \Delta|k$, contradicting our assumption that $A \not \in \Delta|k$.

Let $s_0,\dots,s_n$ be the run of $\Acal_a$ associated with input
$b_1\dots b_n$. Let $St(x_m)=s_m$. As $L(\Acal_a)=L_a(S)$,
we know that $s_0,\dots,s_n$ is an accepting run. This gives us $St(x_0)\in I,
St(x_{m-1})\trans{b_m}_\Acal St(x_m)$ and $St(x_n)\in F$. It remains to show that
$St(x_m):A\in\Gamma'|x_m$. We will do so by induction on $m$.

Base case: As $\odia a A\in\Gamma'|x_0$, by $\Acal$-propagation
we obtain $s_0:A\in\Gamma'|x_0$.

Inductive case: Suppose $x_m\trans{b_{m+1}}_{\Pcal(j,k)} x_{m+1}$. By the inductive
hypothesis, $s_m:A\in\Gamma'|x_m.$ There are two cases to consider:
\begin{itemize}
\item The transition $x_m \trans{b_{m+1}}_{\Pcal(j,k)} x_{m+1}$ also exists
in $\Rcal(\Gamma',j,k)$. In this case, by $\Acal$-propagation, we have $s_{m+1} : A \in \Gamma'|x_{m+1}.$
\item The transition $x_m \trans{b_{m+1}}_{\Pcal(j,k)} x_{m+1}$ is obtained from $\Rcal(\Gamma',j,k)$ through 
the identification of unrealised leaf nodes with their closest ancestors. There are two subcases:
\begin{itemize}
\item $x_m = \lambda(y)$ for some unrealised leaf node $y$ such
that $\Gamma'|x_m = \Gamma'|y$, and $y \trans{b_{m+1}}_{\Rcal(\Gamma',j,k)} x_{m+1}.$
Since $\Gamma'|x_m = \Gamma'|y$, we have that $s_m : A \in \Gamma'|y$ and
it follows by $\Acal$-propagation that $s_{m+1} : A \in \Gamma' | x_{m+1}.$
\item $x_{m+1} = \lambda(y)$ for some unrealised leaf node $y$ such that
that $\Gamma'|x_{m+1} = \Gamma'|y$, and $x_m \trans{b_{m+1}}_{\Rcal(\Gamma',j,k)} y.$
By $\Acal$-propagation, $s_{m+1} : A \in \Gamma' | y = \Gamma' | x_{m+1}.$
\end{itemize}
\end{itemize}

Thus when $Prove(S,\Gamma)$ calls $Prove_2(S,\Gamma,d)$, it will construct an
$S$-stable sequent and terminate.
\end{proof}

\end{document}